\documentclass[letter,11pt]{article}

\usepackage{algorithm}
\usepackage{amsthm}
\usepackage[noend]{algpseudocode}
\usepackage{relsize}
\usepackage{array}
\usepackage{enumitem}
\usepackage{geometry}
\usepackage{graphicx}
\usepackage{mathtools}
\usepackage{subcaption}
\usepackage{tabularx}
\usepackage{tcolorbox}
\usepackage{url}
\usepackage{varwidth}
\usepackage{xspace}
\usepackage{alltt}
\usepackage{balance}

\newcommand{\myparagraph}[1]{\smallskip\noindent {\bf #1.}}

\newtheorem{theorem}{Theorem}[section]
\newtheorem{lemma}{Lemma}[section]

\newtheorem{corollary}{Corollary}[section]

\algrenewcommand\algorithmicindent{1em}

\algnewcommand{\algcomment}[1]{\qquad{\color{purple!40!black}\emph{// #1}}}    % Trailing comment
\algnewcommand{\alglinecomment}[1]{{\color{purple!40!black}\emph{// #1}}}      % Full line comment

\algnewcommand\alglocal{\textbf{local}}                             % local variables
\algnewcommand\algreturn{\textbf{return}}                           % return statement
\algnewcommand\algeach{\textbf{each}}                               % 'each' for loops

\algnewcommand\algto{\textbf{to}}								                  	% 'to' for loops
\algnewcommand\algis{\textbf{is}}                                   % 'is' for conditions
\algnewcommand\algnot{\textbf{not}}                                 % 'not' for conditions
\algnewcommand\algin{\textbf{in}}                                   % 'in' for containers
\algnewcommand\algempty{\textbf{empty}}                             % 'empty' for containers

\algnewcommand\algand{\textbf{and}}                                 % 'and' for conditions
\algnewcommand\algor{\textbf{or}}                                   % 'or' for conditions
\algnewcommand\algassign{\ensuremath{\gets} }                       % local variable assign

\algnewcommand\algtrue{\textbf{true}}                               % true literal
\algnewcommand\algfalse{\textbf{false}}                             % false literal
\algnewcommand\algnull{\ensuremath{\perp}}                          % null literal

\algnewcommand\algarray[1]{\textnormal{array}\ensuremath{\langle}#1\ensuremath{\rangle}}    % array type

\algnewcommand\algorithmicinparallel{\textbf{in parallel}}
\algblockdefx[PARFOR]{ParallelFor}{EndParallelFor}[1]
{\algorithmicfor\ #1\ \algorithmicdo\ \algorithmicinparallel}
{\algorithmicend\ \algorithmicfor}

\makeatletter
\ifthenelse{\equal{\ALG@noend}{t}}%
{\algtext*{EndParallelFor}}
{}%
\makeatother

\algblockdefx[PARFOR]{InParallel}{EndInParallel}[0]
{\algorithmicinparallel\ \algorithmicdo}
{}

\makeatletter
\ifthenelse{\equal{\ALG@noend}{t}}%
{\algtext*{EndInParallel}}
{}%
\makeatother

\mathtoolsset{showonlyrefs}

\newcommand*{\vareps}{\varepsilon}
\DeclareMathOperator{\polylog}{\textnormal{polylog}}

\renewcommand{\log}{\lg}

\newcommand{\bsize}{\ell}

\newcommand{\PRAM}[0]{\ensuremath{\mathsf{PRAM}}}

\newcommand{\crcwpram}{CRCW \ensuremath{\mathsf{PRAM}}}

\newcommand{\multipram}[0]{\emph{multiprefix} CRCW \ensuremath{\mathsf{PRAM}}\renewcommand{\multipram}[0]{multiprefix CRCW \ensuremath{\mathsf{PRAM}}}}

\newcolumntype{P}[1]{>{\centering\arraybackslash}p{#1}}

\begin{document}

  \title{Work-efficient Batch-incremental Minimum Spanning Trees with Applications to the Sliding Window Model}
  
\begin{comment}
  \author{Daniel Anderson}
  \affiliation{\institution{Carnegie Mellon University}}
  \email{dlanders@cs.cmu.edu}
  
  \author{Guy E. Blelloch}
  \affiliation{\institution{Carnegie Mellon University}}
  \email{guyb@cs.cmu.edu}
  
  \author{Kanat Tangwongsan}
  \affiliation{\institution{Mahidol University\\
  International College}}
  \email{kanat.tan@mahidol.edu}
\end{comment}

\author{\normalsize Daniel Anderson \\ \normalsize Carnegie Mellon University \\ \normalsize dlanders@cs.cmu.edu \and \normalsize Guy E. Blelloch \\ \normalsize Carnegie Mellon University \\ \normalsize guyb@cs.cmu.edu \and \normalsize Kanat Tangwongsan \\ \normalsize Mahidol University\\ \normalsize International College \\ \normalsize kanat.tan@mahidol.edu}
\date{}
  \date{}
 
  \maketitle
  
  \begin{abstract}
  Algorithms for dynamically maintaining minimum spanning trees (MSTs) have received much attention in both the parallel and sequential settings. While
  previous work has given optimal algorithms for dense graphs, all existing parallel batch-dynamic algorithms perform polynomial work per update in the worst case for sparse graphs. In this paper, we present the first work-efficient parallel batch-dynamic algorithm for incremental MST, which can insert $\bsize{}$ edges in $O(\bsize{} \log(1+n/\bsize{}))$ work in expectation and $O(\polylog(n))$ span w.h.p. The key
  ingredient of our algorithm is an algorithm for constructing a \emph{compressed path tree} of an edge-weighted tree, which is a smaller tree that contains
  all pairwise heaviest edges between a given set of marked vertices. Using our batch-incremental MST algorithm, we demonstrate a range of applications that become efficiently solvable in parallel in the sliding-window model, such as graph connectivity, approximate MSTs, testing bipartiteness, $k$-certificates, cycle-freeness, and maintaining sparsifiers.
\end{abstract}
  
  \clearpage
  
  \section{Introduction}

Computing the minimum spanning tree (MST) of a weighted undirected graph is a
classic and fundamental problem that has been studied for nearly a century,
going back to early algorithms of Bor\r{u}vka \cite{boruvka1926jistem}, and
Jarn{\'\i}k \cite{jarnik1930jistem} (later rediscovered by Prim
\cite{prim1957shortest} and Dijkstra \cite{dijkstra1959note}), and later, the
perhaps more well-known algorithm of Kruskal \cite{kruskal1956shortest}. The MST
problem is, given a connected weighted undirected graph, to find a set of edges
of minimum total weight that connect every vertex in the graph. More generally,
the minimum spanning forest (MSF) problem is to compute an MST for every connected
component of the graph. The \emph{dynamic} MSF problem is to do so while responding
to edges being inserted into and deleted from the graph . The \emph{incremental}
MSF problem is a special case of the dynamic problem in which edges are only inserted. While most
dynamic data structures handle only a single update at a time, there has also
been work on \emph{batch-dynamic} data structures, which process a batch of
updates in each round. Typically it is assumed that the size of a batch can
vary from round to round. Batch-dynamic data structures have two potential
advantages---they can allow for more parallelism, and they can in some
situations perform less work than processing updates one at a time.

There has been significant interest in parallelizing incremental and dynamic MSF.   Some of this work studies how to implement single updates in parallel~\cite{pawagi1986log,varman1986parallel,varman1988efficient,jung1988parallel,tsin1988handling,das1994n,ferragina1995erew,ferragina1995technique,das1999erew,KPR18}, and some studies batch updates~\cite{pawagi1989parallel,johnson1992optimal, pawagi1993optimal,shen1993parallel,ferragina1994batch,ferragina1996three}.
%None of this work performs less than polynomial work per edge update.
The most recent and best result~\cite{KPR18} requires $O(\sqrt{n} \log n)$ work per update on $n$ vertices, and only allows single edge updates.  All previous results that support batches of edge updates in polylogarithmic time require $\Omega(n\min(\bsize,n))$ work, where $\bsize$ is the size of the batch.   This is very far from the work performed by the best sequential data structures, which is $O(\log n)$ worst-case time for incremental edge insertions~\cite{sleator1983data,AHLT05}, and $O(\frac{\log^4 n}{\log\log n})$ amortized expected time for fully dynamic edge insertions and deletions~\cite{HRW15}.  

In this paper, we start by presenting a parallel data structure for the \emph{batch-incremental MSF problem}.  It is the first such data structure that achieves polylogarithmic work per edge insertion.  The data structure is work efficient with respect to the fastest sequential single-update data structure, and even more efficient for large batch sizes, achieving optimal linear expected work~\cite{KKT95} when inserting all edges as a batch.  The size of a batch can vary from round to round.  Our main contribution is summarized by the following theorem:

\begin{theorem}
  \label{thm:incremental-mst}
  There exists a data structure that maintains the MSF of a weighted undirected graph that can insert a batch of $\bsize$ edges into a graph
  with $n$ vertices in $O\left(\bsize \log\left(1 + \frac{n}{\bsize}\right)\right)$ work in expectation and $O(\log^2(n))$ span w.h.p.\footnote{We say that $g(n) \in O(f(n))$ \emph{with high probability} (w.h.p.)
    if $g(n) \in O(k f(n))$ with probability at least $1 - O(1/n^k)$ for all $k \geq 0$} in the arbitrary-CRCW PRAM.
\end{theorem}

We then use our batch-incremental MSF data structure to develop various data
structures for graph problems in a batch variant of the sliding window model. In
the sliding-window model~\cite{DGIM02}, one keeps a fixed-sized window with
updates adding new updates to the new side of the window and dropping them from
the old side. For incremental updates, each insertion on the new side of the
window does a deletion of the oldest element on the old side. In general, this
can be more difficult than pure incremental algorithms, but not as difficult as
supporting arbitrary deletion in fully dynamic algorithms. This setup has become
popular in modeling an infinite steam of data when there is only bounded memory,
and a desire to ``forget'' old updates in favor of newer ones. There have been
many dozens, perhaps hundreds, of papers using the model in general. Crouch et.
al.~\cite{crouch2013dynamic} have derived several algorithms for graph problems
in this model. The goal for graph algorithms is typically to use only
$\tilde{O}(n)$ memory.

Here we extend the model to allow for rounds of batch updates on the new side of the window, and batch removals from the old side.  Specifically, for graph algorithms we consider batch edge insertions on the new side and batch edge deletions on the far side.  Our results allow for arbitrary interleavings of batch insertions or deletions, and each of arbitrary size.  Matching up equal sized inserts and deletes, of course, gives a fixed sized window, but we do not require this.  Based on our batch-incremental MSF data structure, we are able to efficiently solve a variety of problems in the batch sliding-window model, including connectivity, $k$-certificate, bipartiteness, $(1 + \epsilon)$-MSF, and cycle-freeness. This uses an approach similar to the one of Crouch et. al.~\cite{crouch2013dynamic}, which is based on sequential incremental MSF.  However, beyond using the batch-incremental MSF data structures, we have to augment their data structures in several ways.  For example, since their focus was on memory and not on time, they did not did not discuss efficient implementations of many of their structures.  Since we are interested in parallelism, time is important in our results.   Our results are summarized by the following theorem:

\begin{theorem}
  \label{thm:sliding-window}
  There exist data structures for the batch sliding-window model (batch edge
  insertions on the new side and deletions on the old size) for the problems of
  maintaining connectivity, $k$-certificate, bipartiteness, $(1 +
  \epsilon)$-MSF, cycle-freeness, and $\vareps$-sparsifiers, that all require
  $\tilde{O}(n)$ memory, support batch updates of size $\bsize$ in
  $\tilde{O}(\bsize)$ work and polylogarithmic span, and queries (except for
  sparsifiers) in polylogarithmic work, where $n$ denotes the number of
  vertices.
\end{theorem}

Finally, we note that we can also apply these techniques to the incremental setting, and, using existing
results on incremental graph connectivity~\cite{simsiri2016work}, obtain fast algorithms there as well.
Table~\ref{tab:summary_bounds} gives more specifics on the individual results and compares them to the
existing bounds for parallel dynamic graphs in the incremental and fully dynamic settings.
{
\renewcommand*{\arraystretch}{1.3}
\begin{table*}[t]
\scriptsize
\centering
  \begin{tabular}{l | p{0.25\textwidth} | p{0.25\textwidth} | p{0.25\textwidth} }
    \hline
  & Incremental (This paper)  & Sliding window (This paper) & Fully dynamic (Previous work) \\
  \hline
  Connectivity & $O(\bsize \alpha(n))^*$ (Previous work \cite{simsiri2016work}) & $O(\bsize \log(1 + n/\bsize))^*$ & $O(\bsize \log(n) \log(1 + n/\bsize))^{*, \dagger}$ \cite{acar2019parallel}  \\
  $k$-certificate & $O(k \bsize \alpha(n))^*$ & $O(k\bsize \log(1 + n/\bsize))^*$ & - \\
  Bipartiteness & $O(\bsize \alpha(n))^*$ & $O(\bsize\log(1+n/\bsize))^*$ & - \\
  Cycle-freeness & $O(\bsize \alpha(n))^*$ & $O(\bsize\log(1+n/\bsize))^*$ & - \\
  MSF & $O(\bsize\log(1+n/\bsize))^*$ & $O(\vareps^{-1}\bsize\log
      (n)\log(1+n/\bsize{}))^{*,\ddagger}$ & $O(\bsize n \log\log\log(n)\log(m/n))$ \cite{ferragina1994batch}  \\
  $\varepsilon$-sparsifier & $O(\vareps^{-2}\bsize\log^4 (n) \alpha(n))^*$ & $O(\vareps^{-2}\bsize\log^4 (n) \log(1+n/\bsize{}))^*$ & - \\
  \hline
  \end{tabular}
  \caption{Work bounds for new and known parallel batch-dynamic graph algorithms in the incremental (insert-only), sliding window, and fully dynamic settings. All algorithms run in $O(\polylog(n))$ span and use $O(n\ \polylog(n))$ space. $\bsize{}$ denotes the batch size of updates. Note that the algorithms in the sliding window model are also applicable to the incremental setting, by simply never moving the left endpoint of the window. For large batch sizes $\bsize{}$, these algorithms sometimes achieve better bounds. Some bounds given are randomized ($*$), amortized ($\dagger$), or give $(1+\vareps)$-approximate solutions ($\ddagger$)}\label{tab:summary_bounds}
\end{table*}
}

\subsection{Technical Overview}

The key ingredient in batch-incremental MSF data structure is a data structure for dynamically
producing a \emph{compressed path tree} of an input tree. Given a weighted
tree with some marked vertices, the compressed path tree with respect to the marked vertices is a minimal tree on the marked vertices and some additional ``Steiner vertices'' such that for every pair of marked vertices, the heaviest edge on the path between them is the same in the compressed tree as in the original tree. That is, the compressed path tree 
represents a summary of all possible pairwise heaviest edge queries on the marked vertices. An example of a compressed path tree is shown in Figure~\ref{fig:compressed_path_tree}. More formally, consider the subgraph consisting of the union of the paths between every pair of marked vertices. The compressed path tree is precisely this subgraph but with all of the non-important vertices of degree at most two spliced out. To produce the compressed path tree, we leverage some recent results on parallel batch-dynamic tree contraction and parallel rake-compress trees (RC trees) \cite{acar2020batch}.

\begin{figure}
  \centering
  \begin{subfigure}{0.45\columnwidth}
    \centering
    \includegraphics[width=0.90\textwidth]{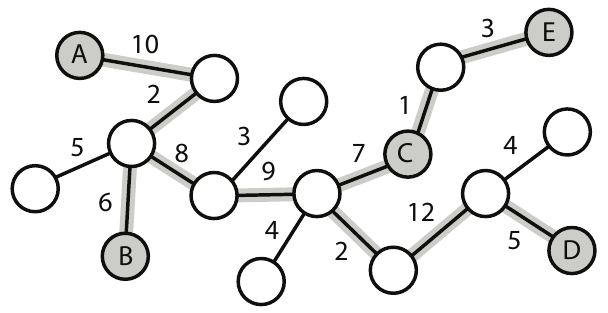}
    \caption{A weighted tree, with some important vertices marked (in gray). The paths between the important vertices are highlighted.}\label{subfig:uncompressed_tree}
  \end{subfigure}\hspace{0.05\columnwidth}
  \begin{subfigure}{0.45\columnwidth}
    \centering
    \includegraphics[width=0.90\textwidth]{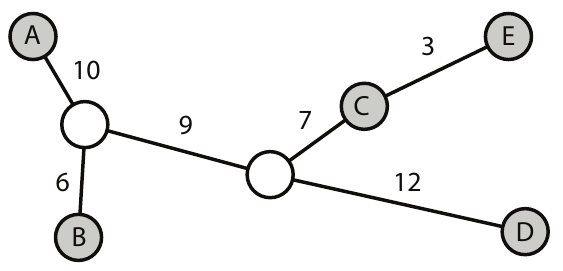}    
    \caption{The corresponding compressed path tree. The edges are weighted to represent the heaviest edge on the corresponding path.}\label{subfig:compressed_path_tree}
  \end{subfigure}
  \caption{A weighted tree and its corresponding compressed path tree with respect to some important vertices.}\label{fig:compressed_path_tree}
\end{figure}

Given a compressed path tree for each component of the graph, our algorithm follows from a generalization of the classic ``cycle rule'' (or ``red rule'') for MSTs, which states that the heaviest edge on a cycle in a graph can not be part of the MST. This fact is used to produce the efficient $O(\log(n))$ time solution
to the sequential incremental MSF problem \cite{sleator1983data}. Our compressed path tree technique generalizes this, in a sense, because it represents
all pairwise paths, and hence all possible cycles, between the newly inserted edges. In the batch setting, our algorithm takes the compressed
path trees and inserts the new batch of edges into them, and computes the MSF of the resulting graph.   For the MSF we can use the algorithm of Cole et. al.~\cite{CKT96}, which is linear work in expectation and logarithmic span w.h.p., which in turn is based on the linear time sequential algorithm~\cite{KKT95}.
Since the compressed path tree is of size $O(\bsize)$, this can be done efficiently. We then
show that the edges selected by this MSF can be added to the MSF of the main graph, and those that were not selected can be removed, resulting
in a correctly updated MSF.

Lastly, using a mix of known reductions and some new ideas, we show how our batch-incremental MSF algorithm
can be used to solve problems in a parallel version of the sliding-window graph streaming model.
A useful ingredient in this step is the \emph{recent edge} property \cite{crouch2013dynamic}, which
says that by weighting the edges of a graph stream with successively lower weights over time, connectivity
between a pair of vertices in the window can be tested by inspecting the heaviest (i.e. oldest) edge
on the path from $u$ to $v$ in an MSF of the graph so far. Combining this idea with the use of several
work-efficient parallel batch-dynamic data structures, we show how to maintain graph connectivity,
bipartiteness, approximate MSFs, $k$-certificates, cycle-freeness, and sparsifiers, subject to
batch updates in $O(\polylog(n))$ work and span per edge update, and $O(n \polylog(n))$ space.

\subsection{Related Work}

MSTs have a long and interesting history. The problem of \emph{dynamically} maintaining the MST under modifications to the underlying graph has been well studied. Spira and Pan \cite{spira1975finding} were the first to tackle the dynamic problem, and give an $O(n)$ sequential algorithm for vertex insertion that is based on Boruvka's algorithm. The first sublinear time algorithm for edge updates was given by Frederickson \cite{frederickson1985data}, who gave an $O(\sqrt{m})$ algorithm. A well-celebrated improvement to Frederickson's algorithm was given by Eppstein et. al \cite{eppstein1997sparsification}, who introduced the \emph{sparsification} technique to reduce the cost to $O(\sqrt{n})$. A great number of subsequent dynamic algorithms, including parallel ones, take advantage of Eppstein's sparsification. The sequential incremental MST problem, i.e., the problem of maintaining the MST subject to new edge insertions but no deletions, can be solved in $O(\log(n))$ time per update using dynamic trees \cite{sleator1983data,AHLT05} to find the heaviest weight edge on the cycle induced by the new edge and remove it.  Holm et. al. gave the first polylogarithmic time algorithm for fully dynamic MST~\cite{HLT01}, supporting updates in $O(\log^4 (n))$ amortized time per operation, and later improved by a $\log \log n$ factor~\cite{HRW15}, still amortized and also in expectation.  No worst case polylogarithmic time algorithm is known for the fully dynamic case.  This paper is concerned with algorithms for MSTs that are both parallel and dynamic. We review a range of existing such algorithms below.

\myparagraph{Parallel single-update algorithms} Pawagi and Ramakrishnan \cite{pawagi1986log} give a parallel algorithm for vertex insertion (with an arbitrary number of adjacent edges) and edge-weight updates in $O(\log(n))$ span but $O(n^2 \log(n))$ work. Varman and Doshi \cite{varman1986parallel,varman1988efficient} improve this to $O(n \log(n))$ work. Jung and Mehlhorn \cite{jung1988parallel} give an algorithm for vertex insertion in $O(\log(n))$ span, and $O(n)$ work. While this bound is optimal for dense insertions, i.e.\ inserting a vertex adjacent to $\Theta(n)$ edges, it is inefficient for sparse graphs.

Tsin \cite{tsin1988handling} show how to extend the work of Pawagi  and Ramakrishnan \cite{pawagi1986log} to handle vertex deletions in the same time bounds, thus giving a fully-vertex-dynamic algorithm parallel algorithm in $O(\log(n))$ span and $O(n^2 \log(n))$ work. Das and Ferragina \cite{das1994n} give algorithms for inserting and deleting edges in $O(\log(m/n)\log(n))$ span and $O(n^{2/3}\log(m/n))$ work. Subsequent improvements by Ferragina \cite{ferragina1995erew,ferragina1995technique}, and Das and Ferragina \cite{das1999erew} improve the span bound to $O(\log(n))$ with the same work bound.

\myparagraph{Parallel batch-dynamic algorithms} The above are all algorithms for single vertex or edge updates. To take better advantage of parallelism, some algorithms that process batch updates have been developed.
Pawagi \cite{pawagi1989parallel} gives an algorithm for batch vertex insertion that inserts $\bsize$ vertices in $O(\log(n)\log(\bsize))$ span and $O(n\bsize\log(n) \log(\bsize))$ work. Johnson and Metaxas \cite{johnson1992optimal} give an algorithm for the same problem with $O(\log(n)\log(\bsize))$ span and $O(n\bsize)$ work.

Pawagi and Kaser \cite{pawagi1993optimal} were the first to give parallel batch-dynamic algorithms for fully-dynamic MSTs. For inserting $\bsize$ independent vertices, inserting $\bsize$ edges, or decreasing the cost of $\bsize$ edges, their algorithms takes $O(\log(n)\log(\bsize))$ span and $O(n\bsize)$ work. Their algorithms for increasing the cost of or deleting $\bsize$ edges, or deleting a set of vertices with total degree $\bsize$ take $O(\log(n) + \log^2(\bsize))$ span and $O(n^2 (1 + \frac{\log^2(\bsize)}{\log(n)}))$ work. Shen and Liang \cite{shen1993parallel} give an algorithm that can insert $\bsize$ edges, modify $\bsize$ edges, or delete a vertex of degree $\bsize$ in $O(\log(n)\log(\bsize))$ span and $O(n^2)$ work. Ferragina and Luccio \cite{ferragina1994batch,ferragina1996three} give algorithms for handling $\bsize = O(n)$ edge insertions in $O(\log(n))$ span and $O(n \log\log\log(n) \log(m/n))$ work, and $\bsize$ edge updates in $O(\log(n)\log(m/n))$ span and $O(\bsize n \log\log\log(n) \log(m/n))$ work. Lastly, Das and Ferragina's algorithm \cite{das1994n} can be extended to the batch case to handle $\bsize$ edge insertions in $O((\bsize + \log(m/n)\log(n))$ span and $O(n^{2/3}(\bsize + \log(m/n)))$ work.

For a thorough and well written survey on the techniques used in many of the above algorithms, see Das and Ferragina~\cite{das1995parallel}.

\myparagraph{Sliding window dynamic graphs} Dynamic graphs in the sliding window model were studied by Crouch et. al~\cite{crouch2013dynamic}. In
the sliding window model, graphs consist of an infinite stream of edges $\langle e_1, e_2, ... \rangle$, and the goal of queries is to compute
some property of the graph over the edges $\langle e_{t-L+1}, e_{t-L+2}, \dots,
e_t \rangle$, where $t$ is the current time and and $L$ is the fixed length
of the window. Crouch et. al showed that several problems, including
$k$-connectivity, bipartiteness, sparsifiers, spanners, MSFs, and matchings, can be efficiently computed in this model. In particular, several of these results used a data structure for incremental MSF as a key ingredient.
All of these results assumed a sequential model of computation.

%%% Local Variables:
%%% mode: latex
%%% TeX-master: "main"
%%% End:

  \section{Preliminaries}

\subsection{Model of Computation}\label{subsec:model}
\myparagraph{Parallel Random Access Machine (\PRAM{})}
The parallel random access machine (\PRAM{}) model is a classic
parallel model with $p$ processors that work in lock-step, connected
by a parallel shared-memory~\cite{jaja1992introduction}. In this paper
we primarily consider the Concurrent-Read Concurrent-Write model
(\crcwpram{}), where memory locations are allowed to be concurrently
read and concurrently written to. If multiple writers write to the
same location concurrently, we assume that an arbitrary writer wins.
We analyze algorithms in terms of their \emph{work} and \emph{span},
where work is the total number of vertices in the thread DAG and where
span (also called depth) is the length of the longest path in the DAG
\cite{Blelloch96}.

\subsection{Tree Contraction and Rake-compress trees}

Tree contraction produces, from a given input tree, a set of smaller (contracted) trees,
each with a subset of the vertices from the previous one, until the final
layer which has a single remaining vertex. The original technique of Miller
and Reif \cite{miller1989parallel} produces a set of $O(\log(n))$ trees w.h.p,
with a geometrically decreasing number of vertices in each one. Specifically,
the technique of Miller and Reif involves sequential rounds of applying two
operations in parallel across every vertex of the tree, rake and compress.
The rake operation removes a leaf from the tree and merges it with its neighbor.
The compress operation
takes a vertex of degree two whose neighbors are not leaves and removes it,
merging the two adjacent edges. An important detail of Miller and Reif's
algorithm is that it operates on bounded-degree trees. Arbitrary degree
trees can easily be handled by converting them into equivalent bounded
degree trees, which can be done dynamically at no extra cost asymptotically
as is described in \cite{acar2020batch}.

A powerful application of tree contraction is that it can be used to
produce a recursive clustering of the given tree with attractive properties.
Using Miller and Reif's tree contraction, a recursive clustering can be 
produced that consists of $O(n)$ clusters, with recursive height $O(\log(n))$ w.h.p.
Such a clustering can be represented as a so-called
\emph{rake-compress tree} (RC tree) \cite{acar2005experimental}. Recent work 
has shown how to maintain a tree contraction dynamically subject to batch-dynamic
updates, work efficiently, and with low span \cite{acar2020batch}.
These results also facilitate maintaining RC trees subject to batch-dynamic updates
work-efficiently and in low span.
Specifically, an RC tree can be built in $O(n)$ work in expectation
and $O(\log^2(n))$ span w.h.p., and subsequently updated in $O(\bsize \log(1+n/\bsize))$ work
in expectation and $O(\log^2(n))$ span w.h.p.\ for batches of $\bsize$ edges.

\myparagraph{Rake-compress trees} RC trees encode a recursive clustering of a tree.
A cluster is defined to be a connected subset of vertices and edges of the original
tree. Importantly, it is possible for a cluster to contain an edge without containing
its endpoints. The \emph{boundary vertices} of a cluster $C$ are the vertices $v \notin C$
such that an edge $e \in C$ has $v$ as one of its endpoints. All of the clusters
in an RC tree have at most two boundary vertices. A cluster with no boundary vertices
is called a \emph{nullary cluster}, a cluster with one boundary is a
\emph{unary cluster} (corresponding to a rake)
and a cluster with two boundaries is \emph{binary cluster}
(corresponding to a compress). The root cluster is always
a nullary cluster. Nodes in an RC tree correspond to clusters, such that a node is always
the disjoint union of its children. The leaf clusters of the RC tree are the vertices and
edges of the original tree. Note that all non-leaf clusters have exactly one vertex (leaf)
cluster as a child. This vertex is that cluster's \emph{representative} vertex. Clusters
have the useful property that the constituent clusters of a parent cluster $C$ share
a single boundary vertex in common---the representative of $C$, and their remaining
boundary vertices become the boundary vertices of $C$.
See Figure~\ref{fig:rc-tree} for an example of a tree, a recursive clustering, and its
corresponding RC tree. Note that for a disconnected forest, the RC tree algorithm will simply
produce a separate root cluster for each component.

One of the great powers of RC trees is their ability to handle a multitude
of different kinds of queries~\cite{acar2005experimental}. In addition to subtree and
path queries, they can also facilitate many nonlocal queries, such as centers,
medians, and lowest common ancestors, all in $O(\log(n))$ time. In this paper, we will make use of
path queries, which allow us to find, given a pair of vertices $u$ and $v$,
the heaviest edge on the path from $u$ to $v$. We refer the reader to
\cite{acar2020batch} and \cite{acar2005experimental} for a more in-depth explanation
of RC trees and their properties.

\begin{figure}
  \centering
  \begin{subfigure}{0.45\columnwidth}
    \centering
    \includegraphics[width=\columnwidth]{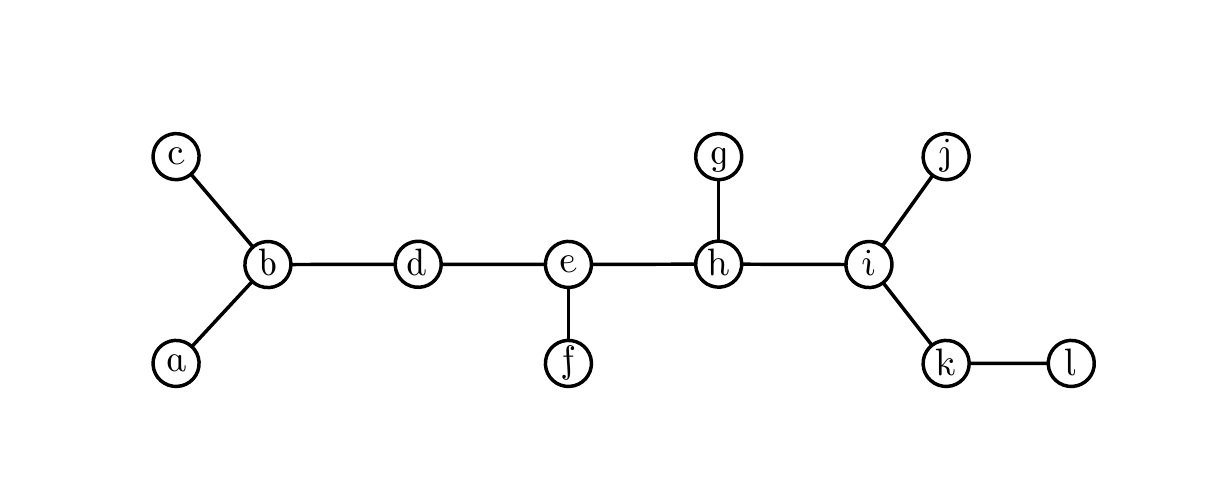}
    \caption{A tree}
  \end{subfigure}
  \begin{subfigure}{0.45\columnwidth}
    \centering
    \includegraphics[width=\columnwidth]{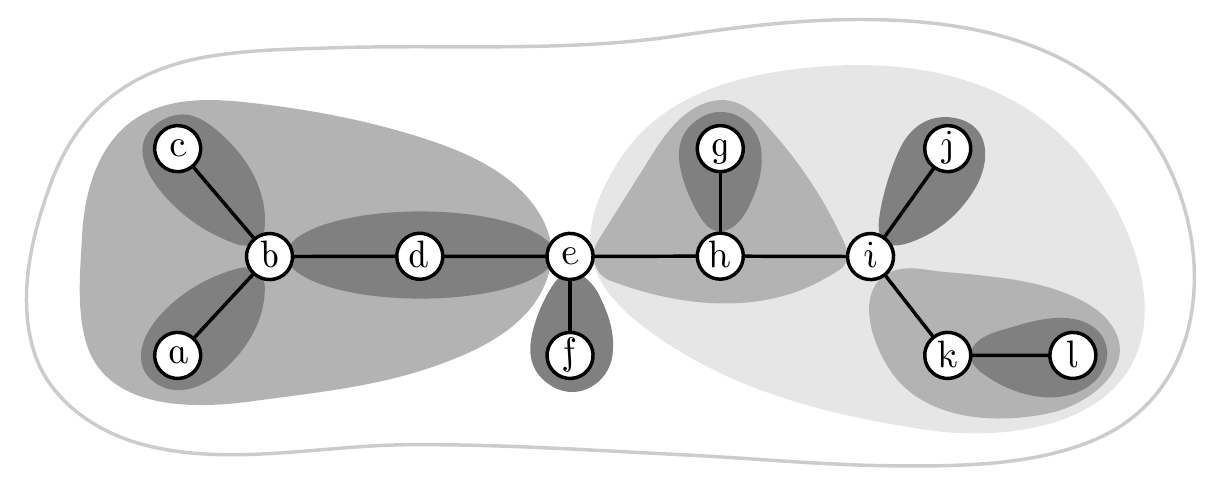}
    \caption{A recursive clustering of the tree produced by tree
      contraction. Clusters produced in earlier rounds are depicted in
      a darker color.}
  \end{subfigure}
  
  \bigskip
  
  \begin{subfigure}{0.9\columnwidth}
    \centering
    \includegraphics[width=0.8\columnwidth]{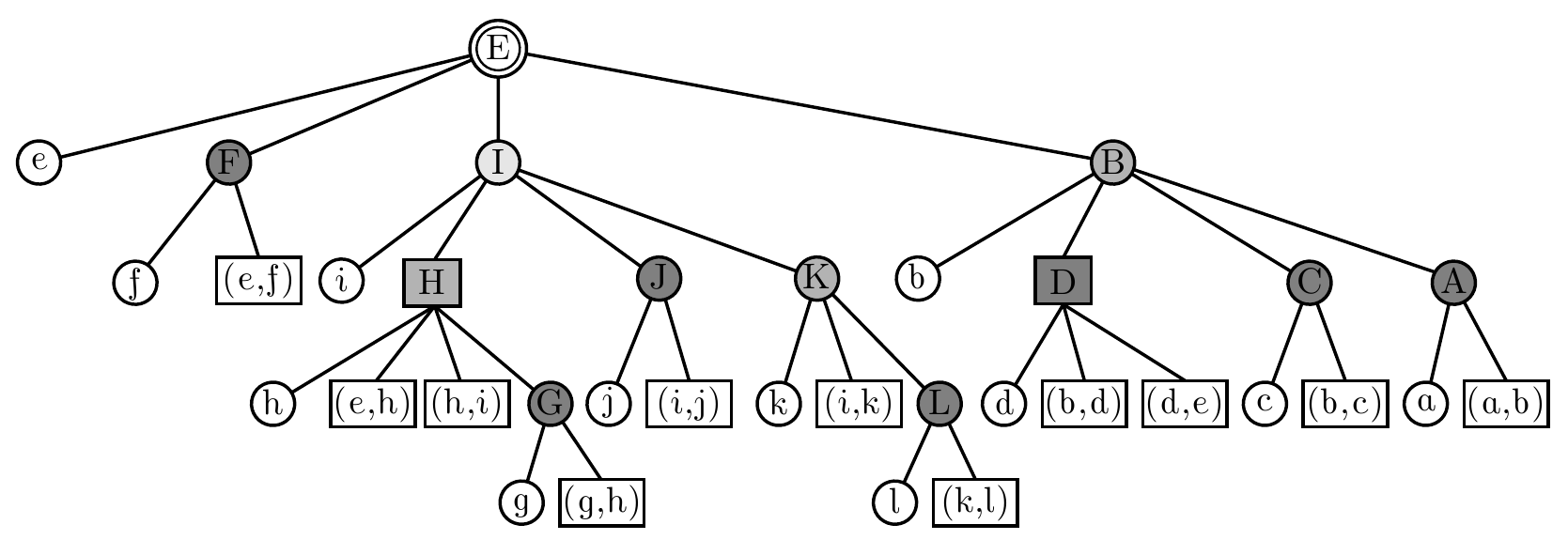}
    \caption{The corresponding RC tree. (Non-base) unary clusters are shown as circles, binary clusters as rectangles, and the finalize (nullary) cluster at the root with two concentric circles. The base clusters (the leaves) are labeled in lowercase, and the composite clusters are labeled with the uppercase of their representative.}
  \end{subfigure}  
  \caption{A tree, a clustering, and the corresponding RC tree \cite{acar2020batch}.}\label{fig:rc-tree}
\end{figure}

  \section{The Compressed Path Tree}

Given an RC tree of height $O(\log(n))$, our algorithm produces the compressed path tree with respect
to $\bsize{}$ marked vertices in $O(\bsize{}\log(1+n/\bsize{}))$ work and $O(\log^2(n))$ span. Note that the RC tree
has height $O(\log(n))$ w.h.p.

Broadly, our algorithm for producing the compressed path tree works as follows.
The algorithm begins by marking the clusters in the RC tree that contain a marked vertex, which
is achieved by a simple bottom-up traversal of the tree.   It then traverses the clusters
of the RC tree in a recursive top-down manner.   When the algorithm encounters a cluster that
contains no marked vertices, instead of recursing further, it can simply generate
a compressed representation of the contents of that cluster immediately. The algorithm uses
the following key recursive subroutine.

\newcommand{\algtype}[1]{{\itshape\bfseries{#1}}}

\begin{itemize}[leftmargin=12pt, topsep=1em]
  \item \textproc{ExpandCluster}($C$ : \algtype{Cluster}) : \algtype{Graph}
  
  Return the compressed path tree of the subgraph corresponding to the graph
  $C \cup \textproc{Boundary}(C)$, assuming that the boundary vertices of $C$ are
  marked.
\end{itemize}

\noindent We make use of the following primitives for interacting with the RC tree.
Each of them takes constant time.

\begin{itemize}[leftmargin=12pt, topsep=1em]
  \item \textproc{Boundary}($C$ : \algtype{Cluster}) : \algtype{vertex list}
  
  Given a cluster in the RC tree, return a list of its boundary vertices.
  
  \item \textproc{Children}($C$ : \algtype{Cluster}) : \algtype{Cluster list}
  
  Given a cluster in the RC tree, return a list of its child clusters.
  
  \item \textproc{Representative}($C$ : \algtype{Cluster}) : \algtype{vertex}
  
  Given a non-leaf cluster in the RC tree, return its representative vertex.
  
  \item \textproc{Weight}($B$: \algtype{BinaryCluster}) : \algtype{number}
    
  Given a binary cluster in the RC tree, return the weight of the heaviest edge
  on the path between its two boundary vertices. This quantity can be
  maintained by the RC tree and hence this takes constant time \cite{acar2020batch}.
  
\end{itemize}
  
\noindent Lastly, we use the following primitives for constructing the resulting
compressed path tree.
  
\begin{itemize}[leftmargin=12pt, topsep=1em]
  \item \textproc{SpliceOut}($G$ : \algtype{Graph}, $v$ : \algtype{vertex}) : \algtype{Graph}
  
  If $v$ has degree two in $G$ and is not marked, splice $v$ out by replacing its two incident
  edges with a contracted edge. The weight of the new edge is the heaviest of the two removed
  edges.
  
  \item \textproc{Prune}($G$ : \algtype{Graph}, $v$ : \algtype{vertex}) : \algtype{Graph}
  
  If $v$ has degree two in $G$, return \textproc{SpliceOut}($G$). Otherwise, if $v$ has degree
  one in $G$, with neighbor $u$, and is not marked, remove $v$ and the edge $(u,v)$, and
  return \textproc{SpliceOut}($G', u$), where $G'$ is the graph remaining after removing $v$
  and $(u,v)$.
\end{itemize}

\noindent The intuition behind the \textproc{Prune} primitive is that without it, our algorithm
would sometimes add redundant vertices to the compressed path tree. The proof of
Lemma~\ref{thm:compressed-path-tree-correctness} illuminates the reason for the precise
behavior of \textproc{Prune}. We give pseudocode for \textproc{ExpandCluster} in Algorithm~\ref{alg:compressed_path_tree}.
The compressed path tree of a marked tree is obtained by calling \textproc{ExpandCluster}(root),
where root is the root cluster of the correspondingly marked RC tree. For a disconnected forest,
\textproc{ExpandCluster} can simply be called on the root clusters of each component.

\begin{algorithm}[th]
  \caption{Compressed path tree algorithm}
  \label{alg:compressed_path_tree}
  \begin{algorithmic}[1]
    \State \alglinecomment{Returns a graph $G$, which is represented by a pair of sets $(V,E)$, where $V$ is the vertex set and $E$ is a set of weighted edges. Edges are represented as pairs, the first element of which is the set of endpoints of the edge, and the second of which is the weight}
    \Procedure{ExpandCluster}{$C$ : \algtype{Cluster}}: \algtype{Graph}
      \If{\algnot{} \textproc{Marked}($C$)}
        \State \alglocal{} $V$ \algassign{} \textproc{Boundary}($C$)
        \If{$C$ is a \algtype{BinaryCluster}}
          \State \alglocal{} $e$ \algassign{} $(V, \textproc{Weight}(C))$
          \State\label{code:non-marked-binary} \algreturn{} $(V, \{ e \})$
        \Else
          \State\label{code:non-marked-non-binary} \algreturn{} $(V, \{ \})$
        \EndIf
      \ElsIf{$C$ is a vertex $v$}
        \State\label{code:marked-single-vertex} \algreturn{} $(\{ v \}, \{\})$
      \Else
        \State\label{code:recurse} \alglocal{} $G$ \algassign{} $\bigcup_{c \in \textproc{Children}(C)} \textproc{ExpandCluster}(c)$
        \State\label{code:prune} \algreturn{} \textproc{Prune}($G$, \textproc{Representative}$(C)$)
      \EndIf
    \EndProcedure
  \end{algorithmic}
\end{algorithm}

\subsection{Analysis}

\myparagraph{Correctness} We first argue that our algorithm for producing
the compressed path tree is correct.
  
\begin{lemma}\label{thm:compressed-path-tree-correctness}
  Given a marked tree $T$ and its RC tree, for any cluster $C$, \textproc{ExpandCluster}($C$)
  returns the compressed path tree of the graph $C \cup \textproc{Boundary}(C)$, assuming that the
  boundary vertices of $C$ are marked.
\end{lemma}
\begin{proof}
We proceed by structural induction on the clusters, with the inductive hypothesis
that \textproc{ExpandCluster}($C$) returns the compressed path tree for
the subgraph  $C \cup \textproc{Boundary}(C)$, assuming that, in addition to the
marked vertices of $T$, the boundary vertices of $C$ are marked. First, consider
an unmarked cluster $C$.
\begin{enumerate}[leftmargin=12pt]
  \item If $C$ is a \algtype{NullaryCluster}, then it has no boundary vertices, and
  since no vertices are marked, the compressed path tree should be empty. Line~\ref{code:non-marked-non-binary}
  therefore returns the correct result.
  \item If $C$ is a \algtype{UnaryCluster}, then it has as single marked boundary
  vertex and no other marked vertices. Therefore the compressed path tree consists
  of the just the boundary vertex, so Line~\ref{code:non-marked-non-binary} returns
  the correct result.
  \item If $C$ is a \algtype{BinaryCluster}, the compressed path tree contains
  its endpoints, and an edge between them annotated with the weight of the
  corresponding heaviest edge in the original tree. Line~\ref{code:non-marked-binary}
  returns this.
\end{enumerate}

\noindent Now, suppose that $C$ is a leaf cluster. Since edges cannot be marked, it must
be a cluster corresponding to a single vertex, $v$. Since $v$ is marked, the
compressed path tree just contains $v$, which is returned by Line~\ref{code:marked-single-vertex}.

We now consider the inductive case, where $C$ is a marked cluster that is not a leaf
of the RC tree. Recall the important facts that the boundary vertices of the children of
$C$ consist precisely of the boundary vertices of $C$ and the representative of $C$, and
that the disjoint union of the children of $C$ is $C$. Using these
two facts and the inductive hypothesis, the graph $G$ (Line~\ref{code:recurse})
is the compressed path tree of the graph $C \cup \textproc{Boundary}(C)$, assuming
that the boundary vertices of $C$ and the representative of $C$ are marked.

It now remains to prove that the \textproc{Prune} procedure (Line~\ref{code:prune})
gives the correct result, i.e., it should produce the compressed path tree without the
assumption that \textproc{Representative}($C$) is necessarily marked. Recall that the
compressed path tree is characterized by having no unmarked vertices of degree less
than three. If \textproc{Representative}($C$) is marked, or if \textproc{Representative}($C$)
has degree at least three, then \textproc{Prune} does nothing, which is correct.
Suppose \textproc{Representative}($C$) has degree two and is unmarked.
\textproc{Prune} will splice out this vertex and combine its adjacent edges.
Observe that splicing out a vertex does not change the degree of any
other vertex in the tree. By the inductive hypothesis, no other vertex of $G$
(Line~\ref{code:recurse}) was unmarked and had degree less than three, hence the
result of Line~\ref{code:prune} is the correct compressed path tree. Lastly, consider
the case where \textproc{Representative}($C$) has degree one and is not marked. \textproc{Prune}
will correctly remove it from the tree, but this will change the degree of its neighboring
vertex by one. If the neighbor was marked or had degree at least four, then it correctly
remains in the tree. If the neighbor had degree three and was not marked, then it will now have
degree two, and hence should be spliced out. As before, this does not change the degree of
any other vertex in the tree, and hence is correct. By the inductive hypothesis, the neighbor
cannot have had degree less than three and been unmarked before calling \textproc{Prune}. Therefore,
in all cases, Line~\ref{code:prune} returns the correct compressed path tree.

By induction on the clusters, we can conclude that the algorithm returns the
compressed path tree of the graph $C \cup \textproc{Boundary}(C)$, assuming that
the boundary vertices of $C$ are marked.
\end{proof}

\begin{theorem}
Given a marked tree $T$ and its RC tree, \textproc{ExpandCluster}(root), where root is the
root of the RC tree, produces the compressed path tree of $T$ with respect to the marked vertices.
\end{theorem}

\begin{proof}
This follows from Lemma~\ref{thm:compressed-path-tree-correctness} and the fact that
the root cluster is a nullary cluster and hence has no boundary vertices.
\end{proof}

\myparagraph{Efficiency} We now show that the compressed path tree can be computed
efficiently.

\begin{lemma}\label{lemma:cpt-is-small}
  A compressed path tree for $\bsize{}$ marked vertices has at most $O(\bsize{})$ vertices.
\end{lemma}

\begin{proof}
  Since a compressed path tree has no non-marked leaves, the compressed path tree
  has at most $\bsize{}$ leaves. Similarly, by definition, the compressed path tree has at
  most $\bsize{}$ internal nodes of degree at most two. The result then follows from the
  fact that a tree with $\bsize{}$ leaves and no internal nodes of degree less than two
  has $O(\bsize{})$ vertices.
\end{proof}

\begin{lemma}[\cite{acar2020batch}]\label{lem:root-to-leaf-paths}
Given the RC tree of a tree $T$ with $n$ vertices, $\bsize{}$ root-to-leaf paths in
the RC tree touch at most $O(\bsize{} \log(1+n/\bsize{}))$ nodes w.h.p.
\end{lemma}

\begin{theorem}\label{thm:compressed-path-tree-speed}
  Given a tree on $n$ vertices and its RC tree, the compressed path tree for $\bsize{}$ marked vertices
  can be produced in $O(\bsize{} \log(1+n/\bsize{}))$ work in expectation and $O(\log(n))$ span w.h.p. on the \crcwpram{}.
\end{theorem}

\begin{proof}
  The algorithm for producing the compressed path tree consists of two bottom-up traversals of the RC
  tree from the $\bsize{}$ marked vertices to mark and unmark the clusters, and a top-down traversal of the same paths
  in the tree. Non-marked paths in the RC tree are only visited if their parent is marked, and since the RC
  tree has constant degree, work performed here can be charged to the parent. Also due to the constant degree
  of the RC tree, at each node during the traversal, the algorithm performs a constant number of recursive calls.
  Assuming that Lines \ref{code:recurse} and \ref{code:prune} can be performed in constant time (to be shown),
  Lemma~\ref{lem:root-to-leaf-paths} implies the work bound of $O(\bsize{} \log(1 + n/\bsize{}))$ in expectation.
  
  To perform Line~\ref{code:recurse} in constant time, our algorithm can perform the set union of the vertex
  set lazily. That is, first run the algorithm to determine the sets of vertices generated by all of the base
  cases, and then flatten these into a single set by making another traversal of the tree. Duplicates can be
  avoided by noticing that the only duplicate in a union of clusters is the representative of their parent cluster.
  Line~\ref{code:prune} can be performed by maintaining the edge set as an adjacency list. Since the underlying tree is
  always converted to a bounded-degree equivalent by the RC tree, the adjacency list can be modified in constant time.
  
  The span bound follows from the fact that the RC tree has height $O(\log(n))$ w.h.p.\ and that
  each recursive call takes constant time.
  
  Lastly, note that this argument also holds for disconnected graphs by simply traversing each component (i.e.\ each root cluster)
  in parallel after the marking phase.
\end{proof}

  \section{Parallel Batch-incremental MSF}

Armed with the compressed path tree, our algorithm for batch-incremental MSF
is a natural generalization of the standard sequential algorithm. The sequential
algorithm for batch-incremental MSF consists in using a dynamic tree data
structure \cite{sleator1983data} to find the heaviest edge on the cycle created
by the newly inserted edge. The classic ``red rule'' says that this edge
can then safely be deleted to obtain the new MSF.

In the batch setting, when multiple new edges are added, many cycles may
be formed, but the same idea still applies. Broadly, our algorithm takes
the batch of edges and produces the compressed path trees with respect to
all of their endpoints. The key observation here is that the compressed
path trees will represent all of the possible paths between the edge endpoints,
and hence, all possible cycles that could be formed by their inclusion. Finding
the new edges of the MSF is then a matter of computing the MSF of the compressed
path trees with the newly inserted edges, and taking those that made it in. The
edges to be removed from the MSF are those corresponding to the compressed path
tree edges that were not part of the MSF.

We express the algorithm in pseudocode in Algorithm~\ref{alg:incremental-mst}.
It takes as input, an RC tree of the current MSF, and the new batch of edges
to insert, and modifies the RC tree to represent the new MSF. The subroutine
\textproc{CompressedPathTrees} computes the compressed path trees for all
components containing a marked vertex in $K$, using Algorithm~\ref{alg:compressed_path_tree}.

\begin{algorithm}[H]
  \caption{Batch-incremental MSF}
  \label{alg:incremental-mst}
  \begin{algorithmic}[1]
    \Procedure{BatchInsert}{RC : \algtype{RCTree}, $E^+$ : \algtype{edge list}}
      \State\label{code:collect-endpoints} \alglocal{} $K$ \algassign{} $\bigcup_{(u,v) \in E^+} \{u,v\}$
      \State\label{code:get-cpt} \alglocal{} $C$ \algassign{} \textproc{CompressedPathTrees}(RC, $K$)
      \State\label{code:compute-mst} \alglocal{} $M$ \algassign{} \textproc{MSF}($C \cup\ E^+$)
      \State\label{code:rc-delete} RC.\textproc{BatchDelete}($E(C) \setminus E(M)$) 
      \State\label{code:rc-insert} RC.\textproc{BatchInsert}($E(M) \cap E^+$)
    \EndProcedure
  \end{algorithmic}
\end{algorithm}

\subsection{Analysis}

\myparagraph{Correctness} We first argue that our algorithm for updating the MSF is correct. We will invoke a classic staple of MST algorithms and their analysis, the ``cycle rule'', also called the ``red rule'' by Tarjan.

\begin{lemma}[Red rule \cite{tarjan1983data}]
  For any cycle $C$ in a graph, and a heaviest edge $e$ on that cycle, there exists a minimum spanning forest of $G$ not containing $e$.
\end{lemma}

\begin{theorem}\label{thm:incremental-mst-correctness}
  Let $G$ be a connected graph. Given a set of edges $E^+$, let $C$ be the compressed path tree of $G$
  with respect to the endpoints of $E^+$, and let $M$ be the MST of $C \cup E^+$. Then a valid MST of $G \cup E^+$ is
  \begin{equation*}
  M' = \textproc{MST}(G) \cup (E(M) \cap E^+) \setminus (E(C) \setminus E(M)),
  \end{equation*}
  where the edges of $C$ are identified with their corresponding heaviest edges in $G$ whose weight they are labeled with.
\end{theorem}

\begin{proof}
   First, observe by some simple set algebra that since $E(M) \cap E^+ = E^+ \setminus (E^+ \setminus E(M))$, we have
   \begin{equation}
   M' = (\textproc{MST}(G) \cup E^+) \setminus (E(C) \setminus E(M)) \setminus (E^+ \setminus E(M)).
   \end{equation}
   We will prove the result using the following strategy. We will begin with the graph $\textproc{MST}(G) \cup E^+$,
   and then show, using the cycle rule, that we can remove all of the edges in $E(C) \setminus E(M)$ and
   $E^+ \setminus E(M)$ and still contain a valid MST. We will then finally show that $M'$ contains no cycles
   and hence is an MST.
   
   Let $e = (u,v)$ be an edge in $E(C) \setminus E(M)$. We want to show that $e$ is a heaviest edge
   on a cycle in $C \cup E^+$. To do so, consider the cycle formed by inserting $e$ into $M$.
   If $e$ was not a heaviest edge on the cycle, then we could replace the heavier edge with $e$
   in $M$ and reduce its weight, which would contradict $M$ being an MST. Therefore $e$
   is a heaviest edge on a cycle in $C \cup E^+$. Since every edge in $C$
   represents a corresponding heaviest edge on a path in $G$, $e$ must also correspond to a heaviest edge
   on the corresponding cycle in $G \cup E^+$. Since $e$ is a heaviest edge on some cycle of
   $G \cup E^+$, the cycle rule says that it can be safely removed. Since we never remove an
   edge in $M$, the graph remains connected, and hence we can continue to apply this argument
   to remove every edge in $E(C) \setminus E(M)$, as desired.
   
   The exact same argument also shows that we can remove all of the edges in $E^+ \setminus E(M)$,
   and hence, we can conclude that $M'$ contains an MST. It remains to show, lastly, that $M'$ contains
   no cycles. First, we observe that since $\textproc{MST}(G)$ contains no cycles, any cycles in $M'$
   must contain an edge from $E^+$. It therefore suffices to search for cycles in $C \cup E^+$ since
   $C$ consists of all possible paths between the endpoints of $E^+$. We will show that if a pair
   of edges both cross the same cut of $C \cup E^+$, that at least one
   of them must be contained in $E(C) \setminus E(M)$ or $E^+ \setminus E(M)$.
   
   First, we note that if an edge is contained in $E(C) \setminus E(M)$ or $E^+ \setminus E(M)$, then
   it is contained in $(E(C) \cup E^+) \setminus E(M)$. Suppose there exists two edges that
   cross the same cut of $C \cup E^+$. Since $M$ is
   a tree, it can not have two edges crossing the same cut, and therefore at least one of them
   can not be in $M$, and hence one of them is contained in $(E(C) \cup E^+) \setminus E(M)$. Therefore
   for any pair of edges crossing the same cut, at least one of them is not contained in $M'$, and hence
   every cut of $M'$ has at most one edge crossing it. This implies that $M'$ does not contain any cycles.
   Combined with the fact that $M'$ contains an MST, we can conclude that $M'$ is indeed an MST of $G \cup E^+$
   as desired.
\end{proof}

\begin{corollary}
  Algorithm~\ref{alg:incremental-mst} correctly updates the MSF.
\end{corollary}

\begin{proof}
Theorem~\ref{thm:incremental-mst-correctness} shows that the algorithm is correct for connected graphs. For
disconnected graphs, we simply apply the same argument for each component, and observe that the previously
disconnected components that become connected will be connected by an MSF.
\end{proof}

\myparagraph{Efficiency} We now show that the batch-incremental MSF algorithm achieves low work and span.

\begin{theorem}
Batch insertion of $\bsize{}$ edges into an MSF using Algorithm~\ref{alg:incremental-mst} takes $O(\bsize{}\log(1+n/\bsize{}))$ work
in expectation and $O(\log^2(n))$ span w.h.p.
\end{theorem}

\begin{proof}
Collecting the endpoints of the edges (Line~\ref{code:collect-endpoints}) takes $O(\bsize{})$ work in expectation
and $O(\log(\bsize{}))$ span w.h.p. using a semisort. By Theorem~\ref{thm:compressed-path-tree-speed}, Line~\ref{code:get-cpt}
takes $O(\bsize{}\log(1+n/\bsize{}))$ work in expectation and $O(\log(n))$ span w.h.p. By Lemma~\ref{lemma:cpt-is-small},
the graph $C \cup E^+$ is of size $O(\bsize{})$, and hence by using the MSF algorithm of Cole et. al.~\cite{CKT96},
which runs in linear work in expectation and logarithmic span w.h.p, Line~\ref{code:compute-mst} takes $O(\bsize{})$
work in expectation and $O(\log(\bsize{}))$ span w.h.p. Lastly, since $C \cup E^+$ is of size $O(\bsize{})$, the batch
updates to the RC tree (Lines \ref{code:rc-delete} and \ref{code:rc-insert}) take $O(\bsize{} \log(1+n/\bsize{}))$ work
in expectation and $O(\log^2(n))$ span w.h.p. Lastly, since $O(\log(\bsize{})) = O(\log(n))$, summing these up, we can conclude
that Algorithm~\ref{alg:incremental-mst} takes $O(\bsize{}\log(1+n/\bsize{}))$ work in expectation and $O(\log^2(n))$ span w.h.p.
\end{proof}

  \section{Applications to the Sliding Window Model}
\newcommand*{\cmdRendering}[1]{{\normalfont\textproc{#1}}\xspace}
\newcommand*{\cmdBulkExpire}[0]{\cmdRendering{BatchExpire}}
\newcommand*{\cmdBulkInsert}[0]{\cmdRendering{BatchInsert}}
\newcommand*{\cmdQuery}[1]{\cmdRendering{#1}}
\newcommand*{\dsName}[1]{{\normalfont\textsf{#1}}\xspace}

We apply our batch-incremental MSF algorithm to efficiently
solve a number of graph problems on a sliding window. For each problem, we
present a data structure that implements the following operations to handle the
arrival and departure of edges:
\begin{itemize}[leftmargin=12pt]
\item $\cmdBulkInsert(B:~\text{\algtype{edge list}})$

Insert the set of edges $B$ into the underlying graph.

\item $\cmdBulkExpire(\Delta: \text{\algtype{int}})$ 

Delete the oldest $\Delta$ edges from the underlying graph.

\end{itemize}
Additionally, the data structure provides query operations specific to the
problem. For example, the graph connectivity data structure offers an
\cmdQuery{isConnected} query operation.

This formulation is a natural extension of the sequential sliding-window model.
Traditionally, the sliding-window model~\cite{DGIM02} entails maintaining the
most recent $W$ items, where $W$ is a fixed, prespecified size. Hence, an
explicit expiration operation in not necessary: the arrival of a new item
triggers the expiration of the oldest item in the window, keeping the size
fixed. More recently, there has been interest in keeping track of variable-sized
sliding windows (e.g., all events in the past $11$ minutes). In this work, we
adopt an interface that allows for rounds of batch inserts (to accept new items)
and batch expirations (to evict items from the old side). Notice that
\cmdBulkExpire{} differs from a delete operation in dynamic algorithms in that
it only expects a count, so the user does not need to know the actual items
being expired to call this operation.
Our results allow for arbitrary interleavings of batch insertions or
expirations, and each of arbitrary size. Matching up equal sized inserts and
expirations can be done to keep the window size fixed, if desired.

Small space is a hallmark of streaming algorithms. For insert-only streams, Sun
and Woodruff~\cite{sun2015tight} show a space lower-bound of $\Omega(n)$ words
for connectivity, bipartiteness, MSF, and
cycle-freeness, and $\Omega(kn)$ words for $k$-certificate assuming a word of
size $O(\log n)$ bits. All our results below, which support not only edge
insertions but also expirations, match these lower bounds except for MSF, which
is within a logarithmic factor.

\subsection{Graph Connectivity}
\label{sec:graph-conn}
\newcommand*{\toa}[0]{\ensuremath{\tau}}

% # Connectivity: If an edge arrives at time t, it is assigned
% a weight of -t. This leads to:
%   *Space: O(n)
%   *bulk-insert(t edges): O(tlog(n/t)) work, O(log t) span
%   *query: O(log n)/query and can perform queries in parallel 
%   *expire: constant time, just advance a timestamp.

We begin with the basic problem of sliding-window graph connectivity, which is
to maintain a data structure that permit the users to quickly test whether a
given pair of vertices $u$ and $v$ can reach each other in the graph defined by
the edges of the sliding window. More specifically, we prove the following
theorem:
\begin{theorem}[Connectivity]
  \label{thm:sw-conn}
  For an $n$-vertex graph, there is a data structure, \dsName{SW-Conn}, that
  requires $O(n)$ words of space and supports the following:
  \begin{itemize}[leftmargin=12pt]
  \item $\cmdBulkInsert(B)$ handles $\bsize = |B|$ new edges in $O(1+\bsize\log(n/\bsize))$
    expected work and $O(\log^2 \bsize)$ span whp.
  \item $\cmdBulkExpire(\Delta)$ expires $\Delta$ oldest edges in $O(1)$
    worst-case work and span.
  \item $\cmdQuery{isConnected}(u, v)$ returns whether $u$ and $v$ are connected
    in $O(\log n)$ work and span whp.
  \end{itemize}
\end{theorem}

Following Crouch~et~al.~\cite{crouch2013dynamic}, we will prove this by reducing
it to the problem of incremental minimum spanning tree. Let $\toa(e)$ denote the
position that edge $e$ appears in the stream. Then, implicit in their paper is
the following lemma:
\begin{lemma}[Recent Edge~\cite{crouch2013dynamic}]
  \label{lemma:recent-edge-one}
  If $F$ is a minimum spanning forest (MSF) of the edges in the stream so far,
  where each edge $e$ carries a weight of $-\toa(e)$, then any pair of vertices
  $u$ and $v$ are connected if and only if (1) there is a path between $u$ and
  $v$ in $F$ and (2) the heaviest edge $e^*$ (i.e., the oldest edge) on this
  path satisfies $\toa(e^*) \geq T_W$, where $T_W$ is the $\toa(\cdot)$ of the
  oldest edge in the window.
\end{lemma}

\begin{proof}[Proof~of~Theorem~\ref{thm:sw-conn}]
  We maintain (i)~an incremental MSF data structure from
  Theorem~\ref{thm:incremental-mst} and (ii)~a variable $T_W$, which tracks the
  arrival time $\toa(\cdot)$ of the oldest edge in the window. The operation
  $\cmdBulkInsert(B)$ is handled by performing a batch insert of $\bsize = |B|$
  edges, where an edge $e \in B$ is assigned a weight of $-\toa(e)$. The
  operation $\cmdBulkExpire(\Delta)$ is handled by advancing $T_W$ by $\Delta$.
  The cost of these operations is clearly as claimed.

  The query $\cmdQuery{isConnected}(u, v)$ is answered by querying for the
  heaviest edge on the path between $u$ and $v$ in the data structure maintained
  and applying the conditions in the recent edge lemma
  (Lemma~\ref{lemma:recent-edge-one}). The claimed cost bound follows because
  the MSF is maintained as a dynamic tree data structure that supports path
  queries in $O(\log n)$~\cite{acar2020batch}.
\end{proof}

Often, applications, including those studied here later, depend on an operation
\cmdQuery{numComponents()} that returns the number of connected components in
the graph. It is unclear how to directly use the above algorithm, which uses
lazy deletion, to efficiently support this query. We now describe a variant,
known as \dsName{SW-Conn-Eager}, which supports \cmdQuery{numComponents()} in
$O(1)$ work.

The number of connected components can be computed from the number of edges in
the minimum spanning forest (MSF) that uses only unexpired edges as
\[
  \text{\# of
    components} = n - \text{\# of MSF edges}.
\]
To this end, we modify \dsName{SW-Conn} to additionally keep a parallel
ordered-set data structure $\mathcal{D}$, which stores all unexpired MSF edges
ordered by $\toa(\cdot)$. This is maintained as follows: The $\cmdBulkInsert$
operation causes $F^+$ to be added to the MSF and $F^-$ to be removed from the
MSF. We can adjust $\mathcal{D}$ using cost at most $O(n\log(n/t))$ work and
$O(\log^2 n)$ span~(e.g.,~\cite{BlellochR98treap,BlellochFS16splitjoin}). The
$\cmdBulkExpire$ operation applies \textproc{Split} to find expired edges
(costing $O(\log n)$ work and span) and explicitly deletes these edges from the
MSF (costing expected $O(\Delta\log(n/\Delta))$ work and $O(\log^2 n)$ span
whp.). With these changes, \cmdQuery{numComponents()} is answered by returning
$n - |\mathcal{D}|$ and \dsName{SW-Conn-Eager} has the following cost bounds:

\begin{theorem}[Connectivity With Component Counting]
  \label{thm:sw-conn-eager}
  For an $n$-vertex graph, there is a data structure, \dsName{SW-Conn-Eager},
  that requires $O(n)$ space and supports the following:
  \begin{itemize}[leftmargin=12pt]
  \item $\cmdBulkInsert(B)$ handles $\bsize = |B|$ new edges in $O(1+\bsize\log(n/\bsize))$
    expected work and $O(\log^2 n)$ span whp.
  \item $\cmdBulkExpire(\Delta)$ expires $\Delta$ oldest edges in
    $O(\Delta\log(1 + n/\Delta) + \log n)$ expected work and $O(\log^2 n)$ span
    whp.
  \item $\cmdQuery{isConnected}(u, v)$ returns whether $u$ and $v$ are connected
    in $O(\log n)$ work and span whp.
  \item $\cmdQuery{numComponents}()$ returns the number of connected components
    in $O(1)$ worst-case work and span.
  \end{itemize}
\end{theorem}
\noindent We apply this theorem to solve two problems below.
\subsection{Bipartiteness}
To monitor if a graph is bipartite, we apply a known
reduction~\cite{AhnGM12linearsketch,crouch2013dynamic}: a graph $G$ is bipartite
if and only if its cycle double cover has exactly twice as many connected
components as $G$. A cycle double cover $D(G)$ of $G$ is a graph in which each
vertex $v$ is replaced by two vertices $v_1$ and $v_2$, and each edge $(u, v)$,
by two edges $(u_1, v_2)$ and $(u_2, v_1)$. Hence, $D(G)$ has twice as many
vertices as $G$.

We can track the number of connected components of both the graph in the sliding
window and its double cover by running two parallel instances of
\dsName{SW-Conn-Eager}. Notice the edges of the cycle double cover $D(G)$ can be
managed on the fly during \cmdBulkInsert{} and \cmdBulkExpire{}. Hence, we have
the following:
\begin{theorem}[Bipartite Testing]
  \label{thm:bipartiteness}
  For an $n$-vertex graph, there is a data structure, \dsName{SW-Bipartiteness},
  that requires $O(n)$ space and supports the following:
  \begin{itemize}[leftmargin=12pt]
  \item $\cmdBulkInsert(B)$ handles $\bsize = |B|$ new edges in $O(\bsize\log(1+n/\bsize))$
    expected work and $O(\log^2 n)$ span whp.
  \item $\cmdBulkExpire(\Delta)$ expires $\Delta$ oldest edges in
    $O(\Delta\log(1 + n/\Delta) + \log n)$ expected work and $O(\log^2 n)$ span
    whp.
  \item $\cmdQuery{isBipartite}()$ returns a Boolean indicating whether the
    graph is bipartite in $O(1)$ worst-case work and span.
  \end{itemize}
\end{theorem}

% \begin{small}
% \begin{alltt}
% # Bipartiteness: Store both the MSF of the edges and
%   the MSF of a _cycle double cover_ (contains 2n nodes) and
%   apply~\cite{AhnGM12linearsketch}.
%   *additional idea: need a quick way to count # of connected
%   components using only unexpired edges. This can be done
%   by explicitly deleting expired edges and track the # of
%   components. 
%   *Space: O(n) because there're at most n - 1 edges in the tree
%   *bulk-insert(t edges):O(tklog(n/t)) work, O(klog t log*n) span
%   where the list of active tree edges (ordered by time) can be
%   maintained in a tree structure. Since all new edges are lighter,
%   we can just use a join operation.
%   *query: O(1)
%   *expire (m edges): up to the cost of deleting m edges.
% \end{alltt}
% \end{small}
\subsection{Approximate MSF Weight}
To approximate the weight of the MSF, we propose the following: For this
problem, assume that the edge weights are between $1$ and $n^{O(1)}$. Using
known reductions~\cite{ChazelleRT05mst,AhnGM12linearsketch,crouch2013dynamic},
the weight of the MSF of $G$ can be approximated up to $1+\vareps$ by tracking
the number of connected components in graphs $G_0, G_1, \dots$, where $G_i$ is a
subgraph of $G$ containing all edges with weight \emph{at most} $(1 +
\vareps)^i$. Specifically, the MSF weight is given by
\newcommand*{\numConn}{\textsf{cc}}
\begin{equation}
  (n - \numConn(G_0)) + \sum_{i \geq 1} (\numConn(G_{i-1}) - \numConn(G_i))(1+\vareps)^i,
  \label{eq:approx-mst}
\end{equation}
where $\numConn(G)$ is the number of connected components in graph $G$.

Let $R = O(\vareps^{-1}\log n)$. We maintain $R$ instances of
\dsName{SW-Conn-Eager} $F_1, \dots, F_{R-1}$ corresponding to the connectivity
of $G_0, G_1, \dots, G_{R-1}$. The arrival of $\bsize$ new edges involves
batch-inserting into $R$ \dsName{SW-Conn-Eager} instances in parallel.
Symmetrically, edge expiration is handled by batch-expiring edges in $R$
instances in parallel. Additionally, at the end of each update operation, we
recompute equation \eqref{eq:approx-mst}, which involes $R$ terms and calls to
\cmdQuery{numComponents()}. This recomputation step requires $O(R)$ work and
$O(\log R) = O(\log^2 n)$ span. Overall, we have the following:
\begin{theorem}[Approximate MSF]
  \label{thm:approx-mst}
  Fix $\vareps > 0$. For an $n$-vertex graph, there is a data structure for
  approximate MSF weight that requires $O(\vareps^{-1}n\log n)$ space and
  supports the following:
  \begin{itemize}[leftmargin=1em]
  \item $\cmdBulkInsert(B)$ handles $\bsize = |B|$ new edges in $O(\vareps^{-1}\bsize\log
    n\log(1+n/\bsize{}))$ expected work and $O(\log^2 n)$ span
    whp.
  \item $\cmdBulkExpire(\Delta)$ expires $\Delta$ oldest edges in
    $O(\vareps^{-1}\Delta\log n\log(1 + n/\Delta))$ expected work and $O(\log^2
    n)$ span whp.
  \item $\cmdQuery{weight}()$ returns an $(1+\vareps)$-approximation to the
    weight of the MSF in $O(1)$ worst-case work and span.
  \end{itemize}
\end{theorem}
% \begin{small}
% \begin{alltt}
% # (1+eps) Min Spanning Tree Weight: 
%   *approximation: 1 + eps
%   *key idea: keep F_0, F_1, ..., where F_i is the connectivity
%   data structure for edges with weights at most (1 + eps)^i
%   *space: O(eps^\{-1\}n log n) 
%   *bulk-insert(t edges):
%    ::use bulk-insert from connectivity (above) for all trees
%    ::just need to know how about connected components there are
%      per F_i.
%    ::run reduce.
%    ::total cost O(eps^{-1}tlog nlog(n/t) + eps^{-1}log n) work,
%     O(log t + log(eps^{-1} log n)) span.
%   *query: constant -- just read off the value
%   *expire (m edges): cost of deleting m edges in O(eps^{-1}log n) trees
%   but the trees can be worked on in parallel.
% \end{alltt}
% \end{small}

\subsection{$k$-Certificate and Graph $k$-Connectivity}
For a graph $G$, a pair of vertices $u$ and $v$ are $k$-connected if there are
$k$ edge-disjoint paths connecting $u$ and $v$. Extending this, a graph $G$ is
$k$-connected if all pairs of vertices are $k$-connected. This generalizes the
notion of connectivity, which is $1$-connectivity. To maintain a ``witness'' for
$k$-connectivity, we rely on maximal spanning forest decomposition of order $k$,
which decomposes $G$ into $k$ edge-disjoint spanning forest $F_1, F_2, \dots,
F_k$ such that $F_i$ is a maximal spanning forest of $G\setminus(F_1 \cup F_2
\cup \dots \cup F_{i-1})$. This yields a number of useful properties, notably:
\begin{enumerate}[label=(P\arabic*), leftmargin=*]
\item if $u$ and $v$ are connected in $F_i$, then they are at least
  $i$-connected;
\item $u$ and $v$ are $k$-connected in $F_1 \cup F_2 \cup \dots \cup F_k$ iff.
  they are at least $k$-connected in $G$; and
\item $F_1 \cup F_2 \cup \dots \cup F_k$ is $k$-connected iff. $G$ is at least
  $k$-connected.
\end{enumerate}

Crouch~et~al.~\cite{crouch2013dynamic} show how to maintain such decomposition
on a sliding window. When extended to the batch setting, the steps are as
follows: Let $O_0$ be the new batch of edges $B$. For $i = 1, 2, \dots, k$,
insert $O_{i-1}$ into $F_{i}$ and capture the edges being replaced as $O_i$. Via
known reductions~\cite{crouch2013dynamic,AhnGM12linearsketch}, we have that the
unexpired edges of $F_1 \cup F_2 \cup \dots \cup F_k$ is a $k$-certificate in
the sense of properties (P1)--(P3) above. Additionally, this preserves all cuts
of size at most $k$. We have the following:
\begin{theorem}[$k$-Certificate]
  \label{thm:kcert}
  For an $n$-vertex graph, there is a data structure for $k$-certificate that
  requires $O(kn)$ space and supports the following:
  \begin{itemize}[leftmargin=12pt]
  \item $\cmdBulkInsert(B)$ handles $\bsize = |B|$ new edges in
    $O(k\bsize\log(1+n/\bsize))$ expected work and $O(k\log^2 n)$ span whp.
  \item $\cmdBulkExpire(\Delta)$ expires $\Delta$ oldest edges in
    $O(k\Delta\log(1 + n/\Delta))$ expected work and $O(\log^2 n)$ span whp.
  \item $\cmdQuery{makeCert}()$ returns a $k$-certificate involving at most
    $k(n-1)$ edges in $O(kn)$ work and $O(\log n)$ span.
  \end{itemize}
\end{theorem}

\begin{proof}
  We maintain each $F_i$ using a batch incremental MSF data structure from
  Theorem~\ref{thm:incremental-mst}. To allow eager eviction of expired edges,
  we additionally keep for each $F_i$ a parallel ordered-set data
  structure~(e.g.,~\cite{BlellochR98treap,BlellochFS16splitjoin})
  $\mathcal{D}_i$, which stores all unexpired edges of $F_i$. The operation
  \cmdBulkInsert{} is handled by sequentially working on $i=1, 2, \dots, k$,
  where for each $i$, edges are bulk-inserted into the MSF data structure for
  $F_i$, propagating replaced edges to $F_{i+1}$. The ordered-set data structure
  $\mathcal{D}_i$ can be updated accordingly. Note that the size of the changes
  to $\mathcal{D}_i$ never exceeds $O(\bsize)$. The operation \cmdBulkExpire{}
  involves expiring edges in all $\mathcal{D}_i$'s. Finally, the operation
  \cmdQuery{makeCert} is supported by copying and returning $\cup_{i=1}^k
  \mathcal{D}_i$. Because each $F_i$ is a forest, it has at most $n - 1$ edges,
  for a total of at most $k(n-1)$ edges across $k$ spanning forests.
\end{proof}

Testing whether a graph is $k$-connected appears to be difficult in the
fully-dynamic setting. Sequentially, an algorithm with $O(n\log n)$ time per
update is known~\cite{eppstein1997sparsification}. By contrast, for the
incremental (insert-only) setting, there is a recent algorithm with
$\widetilde{O}(1)$ time per update~\cite{GoranciHT18}. For sliding window, as a
corollary of Theorem~\ref{thm:kcert}, the certificate generated can be used to
test $k$-connectivity via a parallel global min-cut
algorithm~(e.g.,~\cite{GeissmannG18,Ghaffari0T20}). Because there are $O(kn)$
edges, this can be computed in $O(kn\log n + n\log^4 n)$ work and $O(\log^3 n)$
span~\cite{Ghaffari0T20}.

% \begin{small}
% \begin{alltt}
% # k-Connectivity: 
%   *Space: O(kn)
%   *bulk-insert(t edges):O(tklog(n/t)) work, O(klog t) span
%   *query:
%    -is_kconnected(u, v): O(log n); can perform queries in parallel 
%    -edge_conn(u, v): O(log k log n), what's the edge connectivity
%    of u and v? true answer if less than k; otherwise, it will say k.
%   *expire: constant time, just advance a timestamp.
% \end{alltt}
% \end{small}
% \begin{theorem}[$k$-Certificate]
%   \label{thm:kcert}
%   For an $n$-vertex graph, there is a data structure, $k$-\dsName{SW-Certificate},
%   that requires $O(kn)$ space and supports the following operations:
%   \begin{itemize}
%   \item $\cmdBulkInsert(B)$
%   \item $\cmdBulkExpire(\Delta)$
%   \end{itemize}
% \end{theorem}

\subsection{Cycle-freeness}
To monitor whether a graph contains a cycle, we observe that a graph that has no
cycles is a spanning forest. Hence, if $F_1$ is a spanning forest of a graph
$G$, then $G\setminus F_1$ must not have any edges provided that $G$ has no
cycles. To this end, we use the data structure from Theorem~\ref{thm:kcert} with
$k = 2$, though we are not interested in making a certificate. To answer whether
the graph has a cycle, we check to see if $F_2$ is empty, which can be done in
$O(1)$ work and span. Hence, we have the following:

\begin{theorem}[Cycle-freeness]
  \label{thm:cyclefreeness}
  For an $n$-vertex graph, there is a data structure for cycle-freeness that
  requires $O(n)$ space and supports the following:
  \begin{itemize}[leftmargin=12pt]
  \item $\cmdBulkInsert(B)$ handles $\bsize = |B|$ new edges in
    $O(\bsize\log(1+n/\bsize))$ expected work and $O(\log^2 n)$ span whp.
  \item $\cmdBulkExpire(\Delta)$ expires $\Delta$ oldest edges in
    $O(\Delta\log(1 + n/\Delta))$ expected work and $O(\log^2 n)$ span whp.
  \item $\cmdQuery{hasCycle}()$ returns true or false indicating whether the
    graph has a cycle in $O(1)$ work and span.
  \end{itemize}
\end{theorem}

\subsection{Graph Sparsification}
\label{sec:graph-sparsification}
\newcommand{\estp}[0]{\widetilde{p}\xspace}

The graph sparsification problem is to maintain a small, space-bounded subgraph
so that when queried, it can produce a sparsifier of the graph defined by the
edges of the sliding window. An $\vareps$-sparsifier of a graph $G$ is a
weighted graph on the same set of vertices that preserves all cuts of $G$ upto
$1\pm \vareps$ but has only about $O(n\cdot\polylog(n))$ edges. Existing
sparsification algorithms commonly rely on sampling each edge with probability
inversely proportional to that edge's connectivity parameter. Specifically, we
take advantage of the following result:
\begin{theorem}[Fung~et~al.~\cite{FungHHP19}]
  \label{thm:edge-sparsifier-prob}
  Given an undirected, unweighted graph $G$, let $c_e$ denote the edge
  connectivity of the edge $e$. If each edge $e$ is sampled independently with
  probability
  \[
    p_e \geq \min\left(1, \tfrac{253}{c_e\vareps^2}\log^2 n \right)
  \]
  and assigned a weight of $1/{p_e}$, then with high probability, the resulting
  graph is an $\vareps$-sparsifier of $G$.
\end{theorem}

In the context of streaming algorithms, implementing this has an important
challenge: the algorithm has to decide whether to sample/keep an edge before
that edge's connectivity is known.

Our aim is to show that the techniques developed in this paper enable
maintaining an $\vareps$-sparsifier with $O(n\cdot\polylog(n))$ edges in the
batch-parallel sliding-window setting. To keep things simple, the bounds, as
stated, are not optimized for $\polylog$ factors.

We support graph sparsification by combining and adapting existing techniques
for fast streaming connectivity estimation~\cite{GoelKP-corr2012} and sampling
sufficiently many edges at geometric probability
scales~(e.g.,~\cite{AhnGM12linearsketch,crouch2013dynamic}).

The key result for connectivity estimation is as follows: For $i = 1, 2, \dots,
L = O(\log n)$ and $j = 1, 2, \dots, K = O(\log n)$, let $G_i^{(j)}$ denote a
subgraph of $G$, where each edge of $G$ is sampled independently with
probability $1/2^i$ and $G_0^{(j)} = G$. Then, the \emph{level} $L(u, v)$,
defined to be the largest $i$ such that $u$ and $v$ are connected in $G_i^{(j)}$
for all $0 \leq j \leq K$, gives an estimate of $uv$ connectivity:
\begin{lemma}[~\cite{GoelKP-corr2012}]
  \label{lem:connestimate}
  With high probability, for every edge $e$ of $G$,
  $\Theta(s_e/\log n) \leq 2^{L(e)} \leq 2c_e$, where $s_e$ denotes strong
  connectivity and $c_e$ denotes edge connectivity.
\end{lemma}
The same argument also gives $c_e \leq \Theta(2^{L(e)}\log n)$ whp. While we
cannot explicitly store all these $G_i^{(j)}$'s, it suffices to store each
$G_i^{(j)}$ as a \dsName{SW-Conn} data structure (Theorem~\ref{thm:sw-conn}),
requiring a total of $O(K\cdot{}L\cdot n) = O(n\log^2n)$ space.

When an edge $e$ is inserted, if the algorithm were able to determine that
edge's connectivity, it would sample that edge with the right probability
($p_e$) and maintain exactly the edges in the sparsifier. The problem, however,
is that connectivity can change until the query time. Hence, the algorithm has
to decide how to sample/keep an edge without knowing its connectivity.
To this end, we resort to a technique adapted from
Ahn~et~al.~\cite{AhnGM12linearsketch}: Let $H_0$ be the graph defined by the
edges of the sliding window and for $i = 1, 2, \dots, L$, let $H_i \subseteq
H_0$ be obtained by independently sampling each edge of $H_0$ with probability
$1/2^i$. Intuitively, every edge is sampled at many probability scales upon
arrival.

Storing all these $H_i$'s would require too much space. Instead, we argue that
keeping each $H_i$ as $\mathcal{Q}_i$, where $\mathcal{Q}_i$ is a
\dsName{$k$-SW-Certificate} data structure (Theorem~\ref{thm:kcert}) with $k =
O(\tfrac{1}{\vareps^2}\log^3 n)$ is sufficient\footnote{We remark that the
  $\mathcal{Q}_i$ instances themselves contain enough information to estimate
  $c_e$ for all edges, but we do not know how to efficiently estimate them.}.
Maintaining these requires a total of $O(knL) = O(\vareps^{-2}n\log^4 n)$ space.

Ultimately, our algorithm simulates sampling an edge $e$ with probability
$2^{-\lfloor \log_2 \estp_e \rfloor}$, where
\[
  \estp_e = \min\left(1, O(2^{-L(e)}\vareps^{-2}\log^2 n) \right),
\]
which uses an estimate of $2^{L(e)}$ in place of $c_e$. It answers a
\cmdQuery{sparsify} query as follows:
\begin{quote}
  For $e \in \bigcup_{i=1}^L \mathcal{Q}_i$, output $e$ in the sparsifier with
  weight $1/\estp_e$ if $e$ appears in $\mathcal{Q}_{\beta(e)}$, where $\beta(e)
  = \lfloor \log_2 \estp_e \rfloor$.
\end{quote}

We now show that the $\mathcal{Q}_i$'s retain sufficient edges.

\begin{lemma}
  With high probability, an edge $e$ that is sampled into $H_{\beta(e)}$ is
  retained in $\mathcal{Q}_{\beta(e)}$.
\end{lemma}
\begin{proof} Consider an edge $e = \{u, v\}$. There are $c_e$ disjoint paths
  between $u$ and $v$. With high probability, because $c_e \leq \Theta(2^{L(e)}
  \log n)$, the expected number of paths that stay connected in $H_{\beta(e)}$
  is at most $2\estp_e \cdot c_e \leq O(\vareps^{-2}\log^3 n)$. By Chernoff
  bounds, it follows that whp., $e$ has edge connectivity in $H_{\beta(e)}$ at
  most $k = O(\vareps^{-2}\log^3 n)$ for sufficiently large constant. Hence, $e$
  is retained in $\mathcal{Q}_{\beta(e)}$ whp.
\end{proof}

This means that at query time, with high probability, every edge $e$ is sampled
into the sparsifier with probability $2^{-\lfloor \log_2 \estp_e \rfloor} \geq
p_e$, so the resulting graph is an $\vareps$-sparsifier
whp.~(Theorem~\ref{thm:edge-sparsifier-prob}). Moreover, the number of edges in
the sparsifier is, in expectation, at most
\[
  \sum_{e \in E(G)} 2\estp_e = O(\vareps^{-2}\log^3 n)\sum_{e \in E(G)}
  \tfrac1{s_e} = O(\vareps^{-2}n\log^3 n),
\]
where we used Lemma~\ref{lem:connestimate} and the fact that $\sum_e 1/s_e \leq
n - 1$~\cite{BenczurK96,FungHHP19}.

All the ingredients developed so far are combined as follows: The algorithm
maintains a \dsName{SW-Conn} data structure for each $G_i^{(j)}$ and a
\dsName{$k$-SW-Certificate} $\mathcal{Q}_i$ for each $H_i$. The \cmdBulkInsert{}
operation involves inserting the edges into $KL + L$ data structures and the
same number of independent coin flips. The cost is dominated by the cost of
inserting into the $\mathcal{Q}_i$'s, each of which takes
$O(k\bsize\log(1+n/\bsize))$ expected work and $O(k\log^2 n)$ span whp. The
\cmdBulkExpire{} operation involves invoking \cmdBulkExpire{} on all the data
structures maintained; the dominant cost here is expiring edges in the
$\mathcal{Q}_i$'s. Finally, the query operation \cmdQuery{sparsify} involves
considering the edges of $\bigcup_{i=1}^L \mathcal{Q}_i$ in parallel, each
requiring a call to $L(e)$, which can be answered in $O(LK\log n) = O(\log^3 n)$
work and span. In total, this costs $O(nkL\log^3 n) = O(n\polylog(n))$ work and
$O(\polylog(n))$ span. The following theorem summarizes our result for graph
sparsification:

\begin{theorem}[Graph Sparsification]
  \label{thm:sparsification}
  For an $n$-vertex graph, there is a data structure for graph (cut) sparsification
  that requires $O(\vareps^{-2}n\log^4 n)$ space and supports the following:
  \begin{itemize}[leftmargin=1em]
  \item $\cmdBulkInsert(B)$ handles $\bsize = |B|$ new edges in
    $O(\tfrac{1}{\vareps^{2}}\bsize\log(1+\tfrac{n}\bsize{})\log^4 n)$ expected
    work and $O(\vareps^{-2}\log^5 n)$ span whp.
  \item $\cmdBulkExpire(\Delta)$ expires $\Delta$ oldest edges in
    $O(\tfrac{1}{\vareps^{2}}\Delta\log(1+\tfrac{n}\Delta{})\log^4 n)$ expected
    work and $O(\log^2 n)$ span whp.
  \item $\cmdQuery{sparsify}()$ returns an $\vareps$-sparsifier with high
    probability. The sparsifier has $O(\vareps^{-2}n\log^3 n)$ edges and is
    produced in $O(n \polylog(n))$ work and $O(\polylog(n))$ span whp.
  \end{itemize}
\end{theorem}

\subsection{Connection to Batch Incremental}
All applications in this section
(Sections~\ref{sec:graph-conn}--\ref{sec:graph-sparsification}) build on top of
the connectivity data structure presented in Section~\ref{sec:graph-conn}. In
the batch incremental setting, an analog of Theorem~\ref{thm:sw-conn} was given
by Simsiri et~al.~\cite{simsiri2016work}, where \cmdBulkInsert{} takes
$O(\bsize\alpha(n))$ expected work and $O(\polylog(n))$ span and
\cmdQuery{isConnected} takes $O(\alpha(n))$ work and span.

From here, we can derive an analog of Theorem~\ref{thm:sw-conn-eager} using the
following ideas: (i) maintain a component count variable, which is decremented
every time a union successfully joins two previously disconnected components;
and (ii) maintain a list of inserted edges that make up the spanning forest.
This can be implemented as follows: Simsiri et~al.~maintains a union-find data
structure and handles batch insertion by first running a find on the endpoints
of each inserted edge and determining the connected components using a spanning
forest algorithm due to Gazit~\cite{Gazit91}. Notice that the edges that Gazit's
algorithm returns are exactly the new edges for the spanning forest we seek to
maintain and can simply be appended to the list. This yields an analog of
Theorem~\ref{thm:sw-conn-eager}, where \cmdBulkInsert{} still takes
$O(\bsize\alpha(n))$ expected work and $O(\polylog(n))$ span,
\cmdQuery{isConnected} takes $O(\alpha(n))$ work and span, and
\cmdQuery{numComponents} takes $O(1)$ work and span.
Ultimately, this means that replacing Theorems~\ref{thm:sw-conn} and
\ref{thm:sw-conn-eager} with their analogs in each of our applications
effectively replaces the $\log(1 + n/\bsize)$ factor in the work term with a
$\alpha(n)$ term, leading to cost bounds presented in
Table~\ref{tab:summary_bounds}.

%%% Local Variables:
%%% mode: latex
%%% TeX-master: "main"
%%% End:

  \section{Conclusion}
\label{sec:concl}

In this paper, we designed the first work-efficient parallel algorithm for batch-incremental
MSTs. In addition to being work efficient with respect to the sequential algorithm, our algorithm
is asymptotically faster for sufficiently large batch sizes. A key ingredient of our algorithm was
the construction of a compressed path tree---a tree that summarizes the heaviest edges on all
pairwise paths between a set of marked vertices. We demonstrated the usefulness of our algorithm by
applying it to a range of problems in a generalization of the sliding window dynamic graph streaming
model.

Several interesting avenues of future work arise from our results. We are, to the best of our knowledge,
the first to tackle sliding window dynamic graph problems in the parallel setting. Investigating other
algorithms in this setting could lead to a variety of new problems, tools, and solutions.
Given our results in the sliding window model, it seems likely that the results of \cite{simsiri2016work}
in the incremental model can be improved for large batch sizes, to at least match our algorithm.
It would also be interesting to explore other applications of our batch-incremental MST algorithm,
or possibly even the compressed path tree by itself.
Finally, since the span bound of our incremental MST algorithm is bottlenecked by the span of the
RC tree algorithms, designing a faster (i.e.\ $O(\log(n))$ span) RC tree algorithm would improve the span
of the results in this paper. We believe that such an algorithm is possible, but will require
further tools and techniques.

\section*{Acknowledgements}

This work was supported in part by NSF grants CCF-1408940 and CCF-1629444. The authors
would like to thank Ticha Sethapakdi for helping with the figures.

  \bibliographystyle{abbrv}
  \bibliography{ref}
  \balance

\end{document}